\documentclass[english,copyright]{eptcs}
\providecommand{\event}{T. Neary, D. Woods, A.K. Seda and N. Murphy (Eds.):\\ The Complexity of Simple Programs 2008.}
\providecommand{\volume}{1}
\providecommand{\anno}{2009}
\providecommand{\firstpage}{205}
\usepackage{breakurl}

\usepackage[latin1]{inputenc} 
\usepackage[T1]{fontenc} 
\usepackage{babel} 
\usepackage{amssymb,amsmath,amsthm}
\usepackage{color} 
\usepackage{graphicx}

\newcommand\Zset{\mathbb{Z}}
\newcommand\Nset{\mathbb{N}}

\newtheorem{defn}{Definition}
\newtheorem{thm}{Theorem}
\newtheorem{lem}[thm]{Lemma}
\newtheorem{prop}{Proposition}

\newcommand\Aaut{\mathfrak{A}} 
\newcommand\Baut{\mathfrak{B}} 

\newcommand\Tri{\bullet} 

\newcommand\definecal[1] 
{
  \expandafter\newcommand
  \csname #1cal\endcsname
  {\mathcal{#1}}
}
\definecal{C} 
\definecal{I} 

\DeclareMathOperator{\Sup}{Sup} 

\newcommand\pplus{\oplus} 
\newcommand\bigpplus{\bigoplus} 

\newcommand\Back{\mathfrak{B}} 
\newcommand\Part{\mathfrak{P}} 
\newcommand\Col{\mathfrak{C}} 

\def\crouge{\includegraphics{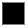}}
\def\cbleu{\includegraphics{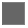}}
\def\cblanc{\includegraphics{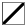}}

\title{A Particular Universal Cellular Automaton}

\author{Nicolas Ollinger \qquad \qquad \qquad \quad Gaétan Richard 
    \email{nicolas.ollinger@lif.univ-mrs.fr \qquad gaetan.richard@lif.univ-mrs.fr}
    \institute{ Laboratoire d'informatique fondamentale de Marseille (LIF),\\
    Aix-marseille Université, CNRS, \\
    39 rue Joliot-Curie, 13013 Marseille, France.
    \thanks{Work supported by a grant of the French ANR}
    }
}

\def\authorrunning{N. Ollinger and G. Richard}
\def\titlerunning{A Particular Universal Cellular Automaton}
\begin{document}
\maketitle

\begin{abstract}
  Signals are a classical tool used in cellular automata constructions
  that proved to be useful for language recognition or firing-squad
  synchronisation. Particles and collisions formalize this idea one
  step further, describing regular nets of colliding signals. In the
  present paper, we investigate the use of particles and collisions
  for constructions involving an infinite number of interacting
  particles. We obtain a high-level construction for a new smallest
  intrinsically universal cellular automaton with 4 states.
\end{abstract}



\section*{Introduction}
\label{sec:intro}

Cellular automata were introduced by
J.~von~Neumann~\cite{Neumann:1966} in the 40s to study
self-reproduction. They consist of a parallel computation model in
discrete time based on an infinite grid of regular \emph{cells}
(endowed with a \emph{state} chosen among a finite alphabet)
interacting synchronously with each other. It is well known that even
simple local interaction can lead to very complex global behavior. One
open question is to find simple local function leading to a global
rule, which is capable either of simulating any Turing machine or
embedding any other cellular automata. Such automata are called
respectively Turing or intrinsically universal. For a detailed survey
of universality of cellular automata, see \cite{Ollinger:2008}.

For dimension 2 or above, the intrinsically universal cellular
automaton constructed by E.~R.~Banks~\cite{Banks:1970} (2 states) is
optimal. For dimension one, the result of M.~Cook~\cite{Cook:2004} (2
states) is optimal for the case of Turing universality whereas the
previous best result on intrinsic universality is with 6 states
\cite{Ollinger:2002}.

All the previously mentioned constructions use structures known as
particles and collisions. Those elements, which are commonly at the
core of algorithmic on cellular automaton constructions, have been
studied in-depth by J.~Mazoyer and V.~Terrier
\cite{Mazoyer-Terrier:1999} . For constructions, a formal approach has
been proposed in \cite{Ollinger-Richard:2007}. With this approach, it
was possible to obtain \cite{Richard:2008} a new high-level proof of
the result of M.~Cook. In this paper, we use this formalism to
construct quite simply an intrinsically universal cellular automaton
with 4 states. The local transition function of the automaton is
constructed to ensure the presence of particles and collisions with a
predictable behavior. The technical part of the proof is to encode
information with those elements to simulate any cellular
automaton. Thus this proof is similar to complexity NP-completeness
proofs.

In section~\ref{sec:ca-univ-self}, we first recall definitions of
cellular automata and different types of universality and then
introduce the formalism used for particles and collisions. In
section~\ref{sec:block}, we construct the rule of the cellular
automaton and show how particles and collisions can be used. In
section~\ref{sec:univ}, we discuss in more details the encoding used
and finish the proof.

\section{Cellular automata, universality and self-organisation}
\label{sec:ca-univ-self}

In this paper, we only consider cellular automata in dimension one,
with the three nearest neighbours. Therefore, a \emph{cellular
  automaton} is a pair $(S,f)$ where $S$ is a finite set of
\emph{states} and $f: S^3 \to S$ is the \emph{local transition
  function}. The automaton acts on a configuration $c \in \S^\Zset$ at
time $t$ by synchronously applying the local transition function
$f(c_{i-1,t},c_{i,t},c_{i+1,t})$ to each cell $c_i$ in $c$ to obtain
the value $c_{i,t+1}$. The result of the global application of $f$ to
$c$ is denoted by $F(c)$.

The automaton acts on $c=(c_i)_{i \in \Zset} \in S^\Zset$
(called \emph{configurations}) by $F(c)_i =
f(c_{i-1},c_i,c_{i+1})$. Moreover, a cellular automaton is
\emph{one-way} if the local function does not depend on its third
argument. A \emph{space-time diagram} $D$ is, in our case, a
bi-infinite sequence of configurations obtained in the dynamics of
cellular automaton. More formally, it can be seen as an element of
$S^{\Zset^2}$ satisfying for all $i,j \in \Zset,
D(i,j+1)=f\left(D(i-1,j),D(i,j),D(i+1,j)\right)$.

Unlike many other known computation systems, cellular automata act on
infinite configurations and do not posses any halting property. These
two properties have led to several different definitions of
universality. One position is to simulate Turing machine by using
``ad-hoc'' properties. A general scheme of different suggested
properties was summarised by B.~Durand and Zs.~R\'{o}ka
\cite{Durand-Roka:1999}. \emph{Turing universality} requires that, for
each Turing machine, there exist a computable function that transforms
each input word of the Turing machine into a ultimately periodic
initial configuration and a pattern such that the initial
configuration eventually containing the pattern if and only if the
Turing machine eventually halts on the input word.

The other approach is to simulate any other cellular automata
evolution. Intuitively, a cellular automaton simulates another
cellular automaton if the space-time diagram of the latter can be
``embedded'' into the former.  Formally, a cellular automaton
$\Aaut=(S,f)$ \emph{simulates} a cellular automaton $\Baut=(T,g)$ if
there exists $m,m',t,t' \in \Nset^+$, $s \in \Zset$ and a surjective
encoding function $e: S^{m} \to T^{m'}$ such that, for any
configuration $c \in T^\Zset$, $\sigma^s(F^t(e^{-1}(c))) \subset
e^{-1}(G^{t'}(c))$ where $\sigma : S^\Zset \to S^\Zset$ is the
\emph{shift} defined for all $c \in S^\Zset$ and $i \in \Nset$ by
$\sigma(c)(i)=c(i+1)$.

\begin{defn}
  A cellular automaton is \emph{intrinsically universal} if it can
  simulate any other cellular automaton.
\end{defn}

The notion of intrinsic universality is stronger than the notion of
Turing universality in the sense that every intrinsic universal
cellular automata is Turing universal but the converse is false. In
the literature, two versions of intrinsic universality exist: one
stronger form has the additional request that encoding be one-one. In
fact most known constructions of universal cellular automata (except
the one of M.~Cook) achieve this stronger definition of
universality. In this paper, we shall not enforce
injectivity. Intuitively, this allows us to have some ``junk'' as long
as it does not interfere with the simulation. Since we are interested
mostly in the regularity of simulation, the two notions are equally
interesting in our case. Furthermore, in a formal point of view, it is
still not known whether they are the same.

As many others constructions, we heavily rely on particles and
collisions to encode and compute information. To define these objects,
we use the approach developed in \cite{Ollinger-Richard:2007} seeing
space-time diagrams as tilings of $S^{\Zset^2}$ with local
constraints. Therefore, we need to introduce some concepts from
discrete geometry: a \emph{coloring} $\Ccal$ is an application from a
subset $Sup(\Ccal) \subset \Zset^2$ to $S$. Such coloring is
\emph{finite} if $Sup(\Ccal)$ is finite. Three natural operations on
colorings are \emph{translation} of a coloring $\Ccal$ by a vector $u
\in \Zset^2$ defined by $(u \cdot \Ccal)(z+u)= \Ccal(z)$;
\emph{disjoint union} of two colorings $\Ccal$ and $\Ccal'$ with
$\Sup(\Ccal) \cap \Sup(\Ccal') = \emptyset$ defined by $\Ccal \pplus
\Ccal'(z)=\Ccal(z) (\textrm{resp. } \Ccal'(z)) $ for all
$z\in\Sup(\Ccal') (\textrm{resp. } \Sup(\Ccal'))$; at last,
\emph{restriction} of a coloring $\Ccal$ to $D \subset \Zset^2$ is
denoted by $\Ccal_{\mid D}$. With those elements, let us give
definitions of the three used structures of self-organisation.

\begin{figure}[htbp]
  \centering
  \begin{tabular}{c p{1cm} c p{1cm} c}
    \includegraphics[height=35pt]{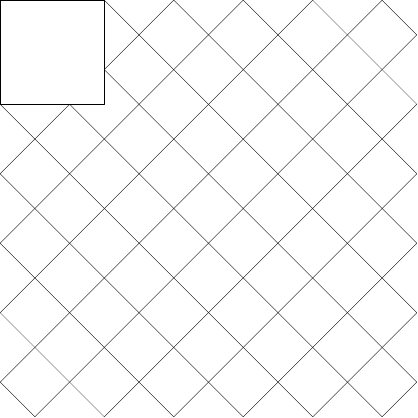} &&
    \includegraphics[height=35pt]{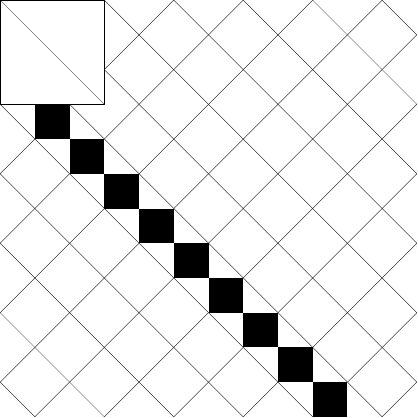} &&
    \includegraphics[height=35pt]{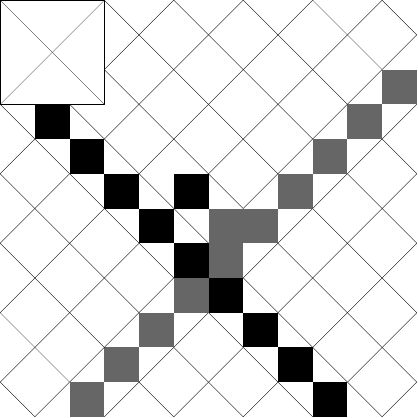} \\
    (a) A background && (b) A particle && (c) A collision \\
  \end{tabular}
  \caption{Examples of structures present in self-organisation}
  \label{fig:Self-org-struct}
\end{figure}

Structures are two-dimensional \emph{backgrounds} (see
Fig.\ref{fig:Self-org-struct}a) which are triplets $\Back=(\Ccal,u,v)$
where $\Ccal$ is a finite coloring and $u,v$ two non-collinear vectors
ensuring that $\bigpplus_{i,j \in \Zset^2} (iu+jv)\cdot\Ccal$ is a
space-time diagram (this space-time diagram is often also referred as
$\Back$ when no confusion is possible). These structures mostly serve
to ``fill'' empty space in constructions. Among backgrounds,
one-dimensional particles travel (see
Fig.\ref{fig:Self-org-struct}b). \emph{Particles} are quadruplets
$\Part = (\Ccal,u,\Back_l,\Back_r)$ where $\Ccal$ is a finite
coloring, $u$ a vector, $\Back_l$ and $\Back_r$ two backgrounds
ensuring that $\Ical = \bigpplus_{k \in \Zset} ku\cdot\Ccal$ separates
the plane in two 4-connected domains $L$ and $R$ ($L$ being the
left-one according to $u$) so that $\Back_{\mid L} \pplus \Ical \pplus
\Back'_{\mid R}$ is a space-time diagram, where $\Back_{\mid L}$
design the restriction of the coloring induced by $\Back$ on
$L$. Particles are often used to convey part of information either
alone or in groups consisting of parallel particles. Last but not
least, a \emph{collision} (see Fig.~\ref{fig:Self-org-struct}c) is a
pair $(\Ccal,L)$ where $\Ccal$ is a finite coloring, $L$ is a finite
sequence of $n$ particles $\Part_i= (\Back_i,\Ccal_i,u_i,\Back'_i)$,
satisfying the following conditions: first, consecutive particles on
the list agree on their common background (for all
$i\in\Zset_n,\quad\Back'_i = \Back_{i+1}$), then particles and finite
perturbation form a star ( $\Ical = \Ccal\pplus\bigpplus_{i\in
  \Zset_n, k \in \Nset} k u_i\cdot\Ccal_i$ cuts the plane in $n$
4-connected zones and For all $i\in\Zset_n$, $\Ccal
\pplus\bigpplus_{k\in\Nset}\left( k u_i\cdot\Ccal_i\pplus k
  u_{i+1}\cdot\Ccal_{i+1}\right)$ cuts the plane in two 4-connected
zones. Let $P_i$ be the one right of $\Part_i$) and at last, $\Col
=\Ical\pplus\bigpplus_{i} {\Back_i}_{| P_i}$ is a space-time diagram.

To describe easily complex space-time diagrams, one idea is to
symbolise particles as lines and collisions as points, giving birth to
a planar map called catenation scheme as the one in upper-left corner
of Fig.~\ref{fig:global}. Formally, a \emph{catenation scheme} is a
planar map whose vertices are labeled by collisions and edges by
particles which are coherent with regards to collisions. Since
catenation schemes are symbolic representations, it is not clear that
there exist associated space-time diagrams. In fact, to go back from a
catenation scheme to a ``real'' space-time diagram, one must give
explicit relative positions of collisions as, for example, by giving
the number of repetitions for each particle (edge) of the scheme. Such
a set of integers is called \emph{affectation} and is \emph{valid} if
the resulting object is a space-time diagram. The main result given
for catenation schemes is that set of valid affectations can be
computed from finite catenation schemes.

\begin{thm}[\cite{Ollinger-Richard:2007}]
  \label{th:cat}
  Given a finite catenation scheme, the set of valid affectations is a
  computable semi-linear set.
\end{thm}

In this paper, the constructed automaton is based on particles and
collisions and thus, heavily relies on the methodology used for
catenation scheme. However, since the particles and collisions are
explicitly constructed, they are very small and posses many good
combinatorial properties. Therefore we do not need the whole power of
catenation scheme and can often give simpler arguments for our
specific case. The rest of the paper is devoted to construct the
automaton and prove it is intrinsically universal. This is done in two
steps: first, we give the automaton and show, using catenation scheme,
that it can somehow ``simulate'' the local behavior of any
automaton. Then, we show how to assemble those local simulations into
a global one.

\section{The automaton and elementary block}
\label{sec:block}

In this section, we present the automaton and show how it can
``somehow'' simulate any local transition function of any one-way
cellular automaton in an \emph{elementary block}. Before going on with
this block, one point to notice is that it is sufficient to simulate
one-way cellular automata to be intrinsically universal. Therefore,
our elementary block is constructed such that it takes three inputs
(one transition function and two arguments) and outputs three values
(one transition function and two ``copies'' of the result).

To provide a good view of the automaton, all needed materials are
depicted in Fig.~\ref{fig:details} and Fig.~\ref{fig:global}. The
first figure gives the local transition function and particles,
whereas the second one gives the catenation scheme along with all
interesting extracts of corresponding space-time diagram. The rest of
this section is devoted to describe and explain the contents of these
two figures. The local transition function is fully depicted on the
top of Fig.~\ref{fig:details}. In fact, it is not a real transition
rule since some cases (denoted by question marks) are not used and
thus can take arbitrary values.


\begin{figure}[!p]
  \begin{center}
  \includegraphics[width=0.9\textwidth,height=0.9\textheight]{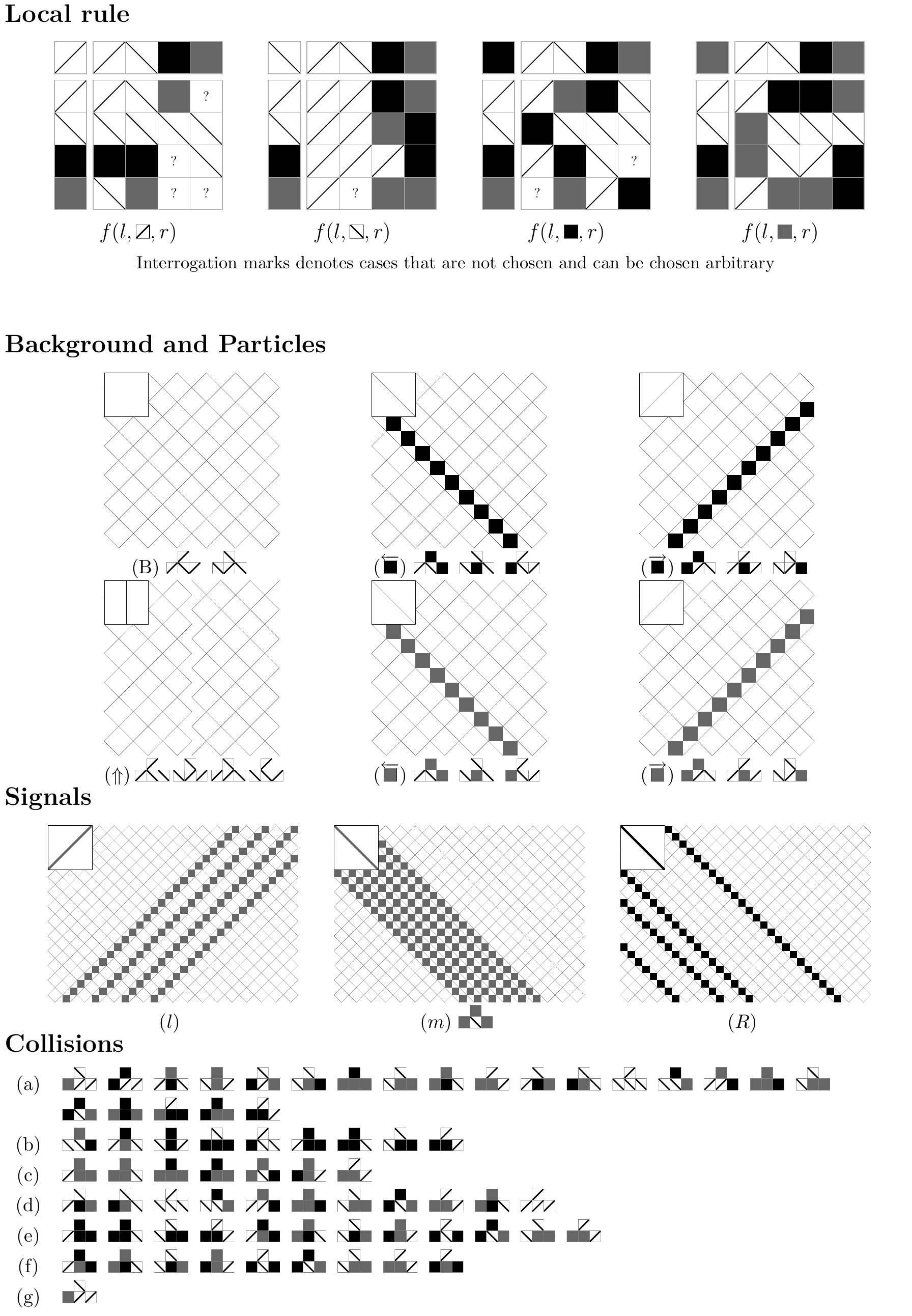} %
  \end{center}
  \caption{Local rule and particles}
  \label{fig:details}
\end{figure}

\begin{figure}[!p]
  \centering
  \includegraphics[width=\textwidth]{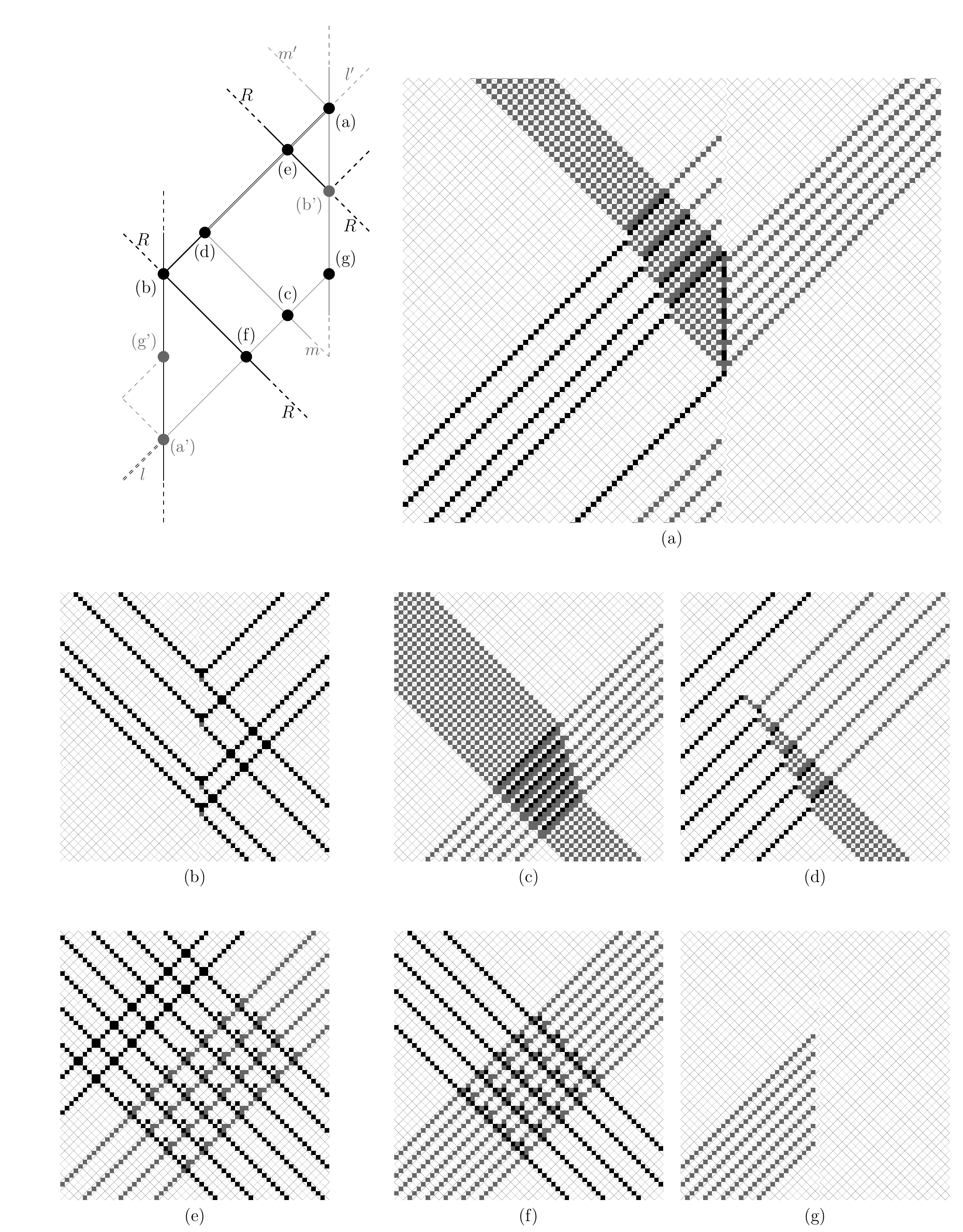}
  \caption{Constructed block and collisions}
  \label{fig:global}
\end{figure}

Even if giving the rule could be sufficient, the next part is devoted
to give intuitions behind local transition function to ease
understanding of the automaton. The first part is devoted to present
background and particles used in the construction. Those structures
are depicted in the corresponding part of Fig.~\ref{fig:details}. For
each structure, we give a meaningful extract of space time, its name
along with used local transition function cases. The formal definition
can be trivially extracted from those diagrams.

One requirement of our cellular automaton is that, contrary to other
known constructions, it does not use a uniform background but a
bi-colored check-board $(B)$. With this new approach, we need two
states (instead of one) for background but this allows us to have a
greater range of particles since each remaining state can lead to two
distinct particles according to its alignment within the
background. Thus, with the two remaining states, we can construct four
different particles $(\overleftarrow{\cbleu}, \overrightarrow{\cbleu},
\overleftarrow{\crouge}, \overrightarrow{\crouge})$. Furthermore,
since background has two different phases, one can construct an
additional particle $(\Uparrow)$ by taking advantage of the gap
between the two phases. Those are the main particles used in the
construction.

To encode information, the basic idea is to use groups of parallel
particles that we call \emph{signals}, depicted in the corresponding
section of Fig.~\ref{fig:details}. In these groups, information can be
encoded in two different ways: either by the number and type of the
particles used or by relative position between those particles. Here,
we use both approaches. As noted earlier, to simulate a on-way
cellular automaton, we need to encode three main types of information:
the \emph{left state} (i.e., state of the left neighbour) is encoded
in unary by the number of (regularly spaced) $\overrightarrow{\cbleu}$
particles (signal $l$). In a similar manner, the \emph{center state}
is encoded in unary by the number of (regularly spaced)
$\overleftarrow{\cbleu}$ particles (signal $m$). The \emph{local
  transition function} is encoded as an array of integers, with each
integer encoded by spaces between a pair of particles
$\overleftarrow{\crouge}$ (signal $R$): the $j$-th element of the
array corresponds to the space between the $j$-th and the $j+1$-th
particle of the signal. To be exact, the space is counted by the sum
of numbers of $\cblanc$ between two consecutive $\crouge$ and two (for
the previously mentioned states). This way of counting may seem a
little obscure but is chosen to give nice formulae in the end of this
section. In the rest of the paper, names of symbols are also used when
speaking of encoded values.

During computation, other kind of signals appear: First, a mirror
image of signal $R$ ($R'$) which encodes the same information using
$\overrightarrow{\crouge}$ particles (going in the opposite
direction). At last, an altered version of $R'$, which we denotes as
$\tilde{R}'$, appears. In this altered version, some of the leftmost
$\overrightarrow{\crouge}$ particles are replaced by
$\overrightarrow{\cbleu}$. Due to choice of background and particles,
our automaton has a special property which makes the proof a lot
easier. Contrary to many other cellular automata, we have no
synchronisation problems.
\pagebreak
\begin{lem}
  When two of the particles (or signals) described above collide,
  occurring collision is always the same.
\end{lem}

\begin{proof}
  As studied by J.~P~Crutchfield {\it et al.}
  \cite{Hordijk-Shalizi-Crutchfield:2001}, the number of ways two
  particles (or signals) can collide depends on the relative position
  between those two particles. Possible relative positions takes into
  account repetition vectors of those particles and constraints
  induced by backgrounds. In our case, repetition vectors are either
  $(1,1)$, $(0,2)$ or $(-1,1)$ which suggests the possibility of two
  distinct collisions. However, the background forbids one of those
  possibility leaving only one possible case.
\end{proof}

With this very strong property, we can focus on symbolic behavior. Let
us go and construct the dynamic. The scheme of elementary block is
depicted in Fig.~\ref{fig:global} (along with extracts of all induced
collisions). The block is made the following way: left and right
borders are delimited by $\Uparrow$ particle. At the bottom, we have a
left state signal $l$, a rule $R$ and a center state signal $m$. The
left state signal is going through the rule (collision $f$) and then
collides with center state signal (collision $c$). This collision
outputs an unused copy of the left signal to the right. This signal is
erased when encountering the right border (collision $g$). Collision
$c$ also sends a signal encoding the the sum of left and center state
to the left. This new signal is encountering (collision $d$) the
mirror copy $R'$ (created when $R$ crossed the left border during
collision $b$). During collision $d$, as many particle of $R'$ are
altered\footnote{Last alteration is an erasing rather than a change of
  the particle but this does not change the behavior} than encoded
value (i.e., the sum of left and center signals). The signal embedding
the sum is destroyed by the collision whereas the altered rule
$\tilde{R}$ proceed to the right. After crossing another rule $R$
(collision $e$), the altered rule $\tilde{R}$ collides with the right
border and produces at the same time a new center state signal and a
left state signal (located right of the border) in collision $a$.

Additional cases needed for each collision in the local transition
function are made explicit at the bottom of
Fig.~\ref{fig:global}. This scheme gives us a symbolic block which
somehow ``computes'' the rule. Now, let us prove that this symbolic
behavior does really correspond to a valid space-time diagram and look
at the details at the computed function. The following proposition
deals with the first problem by ensuring that, under reasonable
conditions, the scheme of symbolic block is a valid space-time
diagram. Moreover, it gives some additional results on regularity of
this block.

\begin{prop}\label{prop:validity}
  For any encoded rule $R$ such that spaces between particles are even
  number greater that four, for any reasonable encoded value in left
  and middle state signal (i.e. both not null and such that their sum
  is less that the number of integers encoded in $R$) the scheme of
  the symbolic block correspond to a valid space-time
  diagram. Moreover, the size of the space time diagram does not
  depend on the encoded states and those blocks can tile the plane.
\end{prop}

\begin{proof}
  First of all, the previous lemma ensures us that there is only one
  type of collision for each pair of particles; non null condition
  ensures that all particles exist. The proof that collisions are
  valid can be directly deduced from extracts given in
  Fig.~\ref{fig:global}. Due to the fact that constraints are local,
  periodicity considerations are sufficient to prove the validity of
  collisions occurring. In fact, it is sufficient to check the
  validity of one repetition of each periodic portion and the joint
  between different portion. The case of collisions $c$, $e$, $f$ and
  $g$ is easily get rid of since the perturbation inside the collision
  is periodic: adding one (or more) particles is the same as
  increasing the size of the periodic portion of the
  collision. Collision $b$ requires space between particles
  $\overleftarrow{\crouge}$ to be greater than four. For collision
  $d$, the constraint is just that there are at least one particles
  $\overrightarrow{\crouge}$ left (i.e. sum of left and right values
  is less that the number of integers encoded in $R$). The last
  significant point is in collision $a$: since the vertical portion in
  the collision has periodicity $(0,4)$, it requires that inter-space
  in $\tilde{R}$ being even (this inter-space is of course the same as
  in $R$). If all these constraints are respected (which is the case
  for our proposition) then the resulting catenation scheme can be
  implemented as a valid space-time diagram. That is, there exists a
  space-time diagram where particles and collisions are positioned in
  the same way as in the catenation scheme.

  Now let us fix $R$, this implies that encoded values in states
  signals are bounded. Therefore, all signals are of bounded size and,
  up to increasing the size of the cell, they can be considered as
  objects with negligible size. With this, it is sufficient to take
  relative positions of signal as in the scheme of
  Fig.~\ref{fig:global} to achieve same size blocks.

  The last point is to prove that such space-time diagrams, associated
  to elementary blocks, tile the plane. First remark is that symbolic
  blocks already tile the plane. Now if we look in details at
  Fig.~\ref{fig:global}, borders are left untouched by all collisions
  and rule signals are only shifted of a constant when encountering a
  border. This implies that crossings of rule signals and border
  (collisions $b$) form a regular grid on the plane. At last, position
  of all other collisions only depends on a small number of neighbour
  points on this grid. It is directly visible for $e$ and
  $a$. Positions of collisions $f$, $g$ and $c$ also depend on $a'$
  but its position is fixed by previous case.  At last, collision $d$
  depends on $c$ which was already treated. 
\end{proof}

With this result, we have an elementary block able to do simple
computation according to a rule $R$ and which can tile the
plane. Before combining those elementary blocks in the next section,
we must first study exactly which function is computed by our
block. For the same periodicity reason than previously, it is
sufficient to look at what happens in the case of collision $a$ in
Fig.~\ref{fig:global}.

\begin{lem}
  Let $R(i)_{1 \leq i \leq N}$ be the array of $N$ integers encoded in
  $R$ and $x_l$ (resp. $x_m$) be the one encoded in $l$ (resp. $m$);
  then the block leaves $R$ unchanged and outputs $m'$ and $l'$ with
  encoded values respectively $R(x_l+x_m) + x_l + x_m - (N + 1)$ and
  $R(x_l + x_m) / 2$. 
\end{lem}

This block computation is alike a cell in a cellular automaton except
that it sends different values for left and right states. This
difference prevents us to give a direct simulation and require quite
additional work to get rid of this problem. A better way to overcome
this problem would be to alter the automaton or use unused cases in
the local transition function to ensure equality of outputs. However,
for the moment, we do not manage to do such a thing. Therefore we
present an alternative (and quite combinatorial) solution to use this
elementary blocks in intrinsic simulation in the next section.

\section{Simulation of cellular automata}
\label{sec:univ}

To achieve intrinsic universality with our cellular automaton, the
idea is to use the constructed elementary blocks and encode
information not directly but only on portions of integers. First step
is to remark that rather than working on only one block, we can work
on chain of two consecutive blocks chained by left output (see
Fig.~\ref{fig:chain}). The new constructed element has three inputs
(namely $x_l$,$x_m$, $x_{\tilde{m}}$) and three outputs $x_{l'}$,
$x_{m'}$ and $x_{\tilde{m}'}$ (not including rule). The key points is
that we see input as integers written in binary and work on bits of
these integers. The basic idea is to simulate a cell of a cellular
automaton using $x_l$ and $x_m$ as inputs, $x_{l'}$ and $x_{m'}$ as
outputs and maintaining the same constant value in all $x_{\tilde{m}}$
and $x_{\tilde{m'}}$.

\begin{figure}[htbp]
  \centering
  \includegraphics{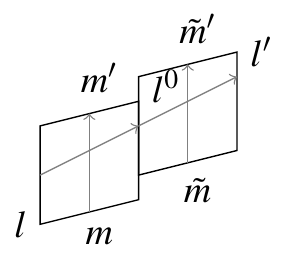}


  %
  %
  \caption{Symbolic chaining of elementary blocks (encoding of state
    is omitted)}
  \label{fig:chain}
\end{figure}

Let us take a one-way cellular automaton $(S,f)$. Since $f$ does not
depends on its third arguments, we can see it has an elements of $S^2
\to S$. Every element of $S$ can be seen as a non null integer and
written as a finite word of fixed size $k$ on binary alphabet denoted
by $s = s_0 s_1 \ldots s_{k-1}$. All values encoded in signals can be
seen as integers with $4 + 8k + 3$ bits (including leading zeros) as
described in table~\ref{tab:coding}. The first four one are called
\emph{header}, the last three one are called \emph{footer} and each
inside block of $8$ (called \emph{digit}) is used to encode one bit of
the state. The size of the rule $N$ is chosen as $2^{4+8k+3}-1$. Only
some parts of the signal for rule $R$ encodes information. Other parts
contain garbage which prevents a one-one encoding.

The method used to encode and decode is depicted in
Tab.~\ref{tab:coding}. The encoding of $l$ and $m$ contains some junk
but is chosen so that $x_l+x_m$ contains the whole information about
states encoded in $l$ and $m$. With this, it is possible to chose the
value of $R(x_l+x_m)$ according to these states. Half of the bits of
$R(x_l+x_m)$ are chosen to ensure the value of $x_{l^0}$. The other
half (denoted by symbol $\Tri$) is used to ensure correct values in
$x_{m'} = R(x_l+x_m) + x_l+ x_m - (N+1)$. In the second cell, one
observe $x_{l^0}+x_{\tilde{m}}$ which has a different header and thus
a different value of any valid $x_l+ x_m$. Therefore, one can chose
$R(x_{l^0}+x_m)$ as $N-l^0+1$ (note that this value is even). With
this choice, $x_{\tilde{m}'}= N - l^0 + 1 + l^0 + \tilde{m} - N - 1 =
x_{\tilde{m}}$ and $x_{l'}$ is on the correct form. With this
encoding, we can simulate any rule of cellular automaton inside our
cells. Leading to the following theorem:

\begin{figure}[htbp]
  \centering
  \begin{tabular}{ l | c  @{ ( } c c c c @{ ) } c  l }
    & header && \multicolumn{2}{c}{digits} && footer \\ 
    \hline
    $N$ & $1111$ & 
    $11$ & $11$ & $11$ & $11$ & 
    $111$ \\
    $x_l$ &  $0101$ & 
    $\bot\bot$ & $s_i0$ & $\bot\bot$ & $00$ &
    $\bot\bot\bot$ \\
    $x_m$ & $0000$ &
    $\bot\bot$ & $00$ & $\bot\bot$ & $s'_i0$ &
    $\bot\bot\bot$ \\
    $x_{\tilde{m}}$ & $1000$ &
    $00$ & $00$ & $00$ & $00$ & 
    $000$ \\
    \hline
    $x_l+x_m$ & $01xx$ &
    $\bot\bot$ & $s_i\bot$ & $\bot\bot$ & $s'_i\bot$ &
    $\bot\bot\bot$ & where $xx=01$ or $10$ \\
    $R(x_l+x_m)$ & $10\Tri\Tri$ &
    $\bar{t}_i1$ & $\Tri\Tri$ & $11$ & $\Tri\Tri$ & 
    $100$ \\
    $x_{m'}$ & $0000$ &
    $\bot\bot$ & $00$ & $\bot\bot$ & $t_i0$ &
    $\bot\bot\bot$ \\ 
    $x_{l^0}$ & $010\bot$ &
    $\bot\bar{t}_i$ & $1\bot$ & $\bot1$ & $1\bot$ &
    $\bot10$ \\
    $N - x_{l^0} +1 $ & $101\bot$ &
    $\bot t_i$ & $0\bot$ & $\bot0$ & $0\bot$  &
    $\bot10$ \\
    $x_{l'}$ & $0101$ &
    $\bot\bot$ & $t_i0$ & $\bot\bot$ & $00$ &
    $\bot\bot\bot$ \\
  \end{tabular}

  $\bot$ denotes arbitrary value and $\Tri$ value to be fixed.
  \caption{Encoding of states inside values}
  \label{tab:coding}
\end{figure}

\begin{thm}
  The cellular automaton presented in Fig~\ref{fig:details} is
  intrinsically universal.
\end{thm} 

\begin{proof}
  With the previously presented encoding, chains of elementary blocks
  can encode any one-way cellular automaton local transition
  rule. Moreover, the encoding ensures conditions for
  proposition~\ref{prop:validity}. Thus constructed elementary blocks
  can tile the plane. If we associate to each cell of the simulated
  cellular automaton space-time diagram the corresponding (up to
  garbage consideration) chain of elementary blocks, we can simulate
  any other automaton. Last point is to note that even if the
  simulation is not horizontal, the problem can be easily overcome
  with standard trick. 
\end{proof}

\section*{Conclusion}
\label{sec:conc}

Our construction heavily relies on signals and manages to give a
construction with two separate levels: local rule is constructed to
ensure a set of particles and collisions. Then, the simulation is made
encoding the computation with this set. This method allows us to have
a clear and understandable construction which separates the local and
global aspect and gives us a new smallest intrinsically universal
cellular automaton.

\bibliographystyle{eptcs}



\end{document}